\documentclass[submission,copyright,creativecommons]{eptcs}
\providecommand{\event}{Th'edu 2019} % Name of the event you are submitting to
\usepackage{breakurl}
\usepackage{hyperref}
\usepackage{graphicx}
\usepackage{rotating}
\usepackage{relsize}
\usepackage{amsmath}
\usepackage{amsthm}
\usepackage{ifthen}
\usepackage[usenames,dvipsnames]{color}
\usepackage{url}
\usepackage{caption}
\usepackage{subcaption}
\usepackage{titlesec}
\usepackage{color}
\usepackage{tikz}
\usepackage[normalem]{ulem}
\usepackage{listings}
\usetikzlibrary{shapes}
\newcommand{\at}[2]{\begin{scope}[shift={#1}]#2\end{scope}}
\newcommand{\smalltext}[1]{\footnotesize{\begin{tabular}{@{}c@{}}#1\end{tabular}}}
\newcommand{\diagnosis}[2]{
   \at{#1}{\draw[dotted] (-1.4,-.4) rectangle (1.4,.4);}
   \draw #1 node {\smalltext{#2}};
}
\newcommand{\decision}[2]{
   \draw #1 ellipse (1 and 0.5);
   \draw #1 node {\smalltext{#2}};
}

\makeatletter
\@ifundefined{lhs2tex.lhs2tex.sty.read}%
  {\@namedef{lhs2tex.lhs2tex.sty.read}{}%
   \newcommand\SkipToFmtEnd{}%
   \newcommand\EndFmtInput{}%
   \long\def\SkipToFmtEnd#1\EndFmtInput{}%
  }\SkipToFmtEnd

\newcommand\ReadOnlyOnce[1]{\@ifundefined{#1}{\@namedef{#1}{}}\SkipToFmtEnd}
\usepackage{amstext}
\usepackage{amssymb}
\usepackage{stmaryrd}
\DeclareFontFamily{OT1}{cmtex}{}
\DeclareFontShape{OT1}{cmtex}{m}{n}
  {<5><6><7><8>cmtex8
   <9>cmtex9
   <10><10.95><12><14.4><17.28><20.74><24.88>cmtex10}{}
\DeclareFontShape{OT1}{cmtex}{m}{it}
  {<-> ssub * cmtt/m/it}{}

\DeclareFontShape{OT1}{cmtt}{bx}{n}
  {<5><6><7><8>cmtt8
   <9>cmbtt9
   <10><10.95><12><14.4><17.28><20.74><24.88>cmbtt10}{}
\DeclareFontShape{OT1}{cmtex}{bx}{n}
  {<-> ssub * cmtt/bx/n}{}
\newcommand{\Conid}[1]{\mathit{#1}}
\newcommand{\Varid}[1]{\mathit{#1}}
\newcommand{\anonymous}{\kern0.06em \vbox{\hrule\@width.5em}}

\renewcommand{\leq}{\leqslant}
\renewcommand{\geq}{\geqslant}
\usepackage{polytable}

\@ifundefined{mathindent}%
  {\newdimen\mathindent\mathindent\leftmargini}%
  {}%

\def\resethooks{%
  \global\let\SaveRestoreHook\empty
  \global\let\ColumnHook\empty}
\newcommand*{\savecolumns}[1][default]%
  {\g@addto@macro\SaveRestoreHook{\savecolumns[#1]}}
\newcommand*{\restorecolumns}[1][default]%
  {\g@addto@macro\SaveRestoreHook{\restorecolumns[#1]}}
\newcommand*{\aligncolumn}[2]%
  {\g@addto@macro\ColumnHook{\column{#1}{#2}}}

\resethooks

\newcommand{\onelinecommentchars}{\quad-{}- }
\newcommand{\commentbeginchars}{\enskip\{-}
\newcommand{\commentendchars}{-\}\enskip}

\newcommand{\visiblecomments}{%
  \let\onelinecomment=\onelinecommentchars
  \let\commentbegin=\commentbeginchars
  \let\commentend=\commentendchars}

\newcommand{\invisiblecomments}{%
  \let\onelinecomment=\empty
  \let\commentbegin=\empty
  \let\commentend=\empty}

\visiblecomments

\newlength{\blanklineskip}
\newcommand{\hsindent}[1]{\quad}% default is fixed indentation
\let\hspre\empty
\let\hspost\empty

\EndFmtInput
\makeatother
\ReadOnlyOnce{polycode.fmt}%
\makeatletter

\newcommand{\hsnewpar}[1]%
  {{\parskip=0pt\parindent=0pt\par\vskip #1\noindent}}

\newcommand{\hscodestyle}{}

\newcommand{\sethscode}[1]%
  {\expandafter\let\expandafter\hscode\csname #1\endcsname
   \expandafter\let\expandafter\endhscode\csname end#1\endcsname}

  {\par\noindent
   \advance\leftskip\mathindent
   \hscodestyle
   \let\\=\@normalcr
   \let\hspre\(\let\hspost\)%
   \pboxed}%
  {\endpboxed\)%
   \par\noindent
   \ignorespacesafterend}

  {\hsnewpar\abovedisplayskip
   \advance\leftskip\mathindent
   \hscodestyle
   \let\hspre\(\let\hspost\)%
   \pboxed}%
  {\endpboxed%
   \hsnewpar\belowdisplayskip
   \ignorespacesafterend}

  {\hsnewpar\abovedisplayskip
   \advance\leftskip\mathindent
   \hscodestyle
   \let\\=\@normalcr
   \(\pboxed}%
  {\endpboxed\)%
   \hsnewpar\belowdisplayskip
   \ignorespacesafterend}

\newcommand{\plainhs}{\sethscode{plainhscode}}

\plainhs

  {\hsnewpar\abovedisplayskip
   \advance\leftskip\mathindent
   \hscodestyle
   \let\\=\@normalcr
   \(\parray}%
  {\endparray\)%
   \hsnewpar\belowdisplayskip
   \ignorespacesafterend}

  {\parray}{\endparray}

  {\(\parray}{\endparray\)}

\def\codeframewidth{\arrayrulewidth}
\RequirePackage{calc}

  {\parskip=\abovedisplayskip\par\noindent
   \hscodestyle
   \arrayrulewidth=\codeframewidth
   \tabular{@{}|p{\linewidth-2\arraycolsep-2\arrayrulewidth-2pt}|@{}}%
   \hline\framedhslinecorrect\\{-1.5ex}%
   \let\endoflinesave=\\
   \let\\=\@normalcr
   \(\pboxed}%
  {\endpboxed\)%
   \framedhslinecorrect\endoflinesave{.5ex}\hline
   \endtabular
   \parskip=\belowdisplayskip\par\noindent
   \ignorespacesafterend}

\newcommand{\framedhslinecorrect}[2]%
  {#1[#2]}

  {\(\def\column##1##2{}%
   \let\>\undefined\let\<\undefined\let\\\undefined
   \newcommand\>[1][]{}\newcommand\<[1][]{}\newcommand\\[1][]{}%
   \def\fromto##1##2##3{##3}%
   }{\) }%

  {\let\orighscode=\hscode
   \let\origendhscode=\endhscode
   \def\endhscode{\def\hscode{\endgroup\def\@currenvir{hscode}\\}\begingroup}
   \orighscode\def\hscode{\endgroup\def\@currenvir{hscode}}}%
  {\origendhscode
   \global\let\hscode=\orighscode
   \global\let\endhscode=\origendhscode}%

\makeatother
\EndFmtInput
\newif\iffinal\finalfalse
  
\newenvironment{itize}{\begin{list}{$\bullet$}{\parsep=0pt\parskip=0pt\topsep=1ex\itemsep=0pt}}{\end{list}}

\newcommand{\ignore}[1]{}

\title{Providing Hints, Next Steps and Feedback \\ in a Tutoring System for Structural Induction}
\author{Josje Lodder  \institute{Faculty of Management, Science and Technology\\ Open University of the Netherlands\\
Heerlen, The Netherlands}  \email{josje.lodder@ou.nl}\and
Bastiaan Heeren 
\institute{Faculty of Management, Science and Technology\\ Open University of the Netherlands\\
Heerlen, The Netherlands}
\email{\quad bastiaan.heeren@ou.nl}
\and
Johan Jeuring 
\institute{Department of Information and Computing Sciences, \\ Universiteit Utrecht, The Netherlands}
\email{\quad j.t.jeuring@uu.nl}
}

\def\titlerunning{A Tutoring System for Structural Induction}
\def\authorrunning{Josje Lodder, Bastiaan Heeren \& Johan Jeuring}

\begin{document}
\maketitle

\begin{abstract}
Structural induction is a proof technique that is widely used to prove statements about discrete structures. Students find it hard to construct inductive proofs, and when learning to construct such proofs, receiving feedback is important. In this paper we discuss the design of a tutoring system, LogInd, that helps students with constructing stepwise inductive proofs by providing hints, next steps and feedback. As far as we know, this is the first tutoring system for structural induction with this functionality. We explain how we use a strategy to construct proofs for a restricted class of problems. This strategy can also be used to complete partial student solutions, and hence to provide hints or next steps. We use constraints to provide feedback. A pilot evaluation with a small group of students shows that LogInd indeed can give hints and next steps in almost all cases.
 
\end{abstract}
\section{Introduction}
%induction in logic
Discrete structures play an important role in many domains, and are foundational for mathematics, logic, and computer science. Examples of such structures are natural numbers, data structures such as lists and trees, but also complex structures such as programming languages. Structural induction is a proof technique that is widely used to prove statements about inductively defined, discrete structures. Mathematical induction can be viewed as a special kind of structural induction, using the inductive definition of the natural numbers as the underlying structure. 

Because discrete structures and structural induction are foundational for mathematics and computer science, they form an integral part of educational programs. For example, proof techniques are part of the ACM Computer Science curriculum.\footnote{\url{https://www.acm.org/binaries/content/assets/education/cs2013_web_final.pdf}} Courses that address proof techniques often require students not only to learn how to prove consequences in a formal system, but also to reason about formal systems, and to independently construct a proof for a statement. A typical example of an exercise that occurs in many textbooks on logic and proof techniques is the following: prove that the number of left parentheses in a logical formula is equal to the number of right parentheses. Such a property can be proved by structural induction, where the structure of the proof follows the structure of the inductive definition of the logical language. Students have to learn this proof technique to construct more fundamental proofs, such as the soundness of a proof system. 
Textbooks and teachers typically instruct students on how to do this, provide some examples, and then let students practice with constructing proofs themselves. As with learning any subject, students need feedback when they are learning how to construct their own proofs~\cite{hattietimperley}. Such feedback can take several forms: it may be about the progress of a student, about recommending a next task to solve, or about the difference between the proof constructed by a student and an expected proof. In this paper we focus on the latter kind of feedback.

This paper discusses the design of a tutoring system for practicing proving statements about inductively defined, discrete structures. Some core features of the system are that it gives feedback on the steps a student takes towards a solution, and hints about which step to take next. We use the knowledge about misconceptions students have to determine what feedback to give. Students have some freedom in setting up their proof, and the system helps in reaching the learning goal of independently constructing a proof for statements about inductively defined discrete structures. Two advantages of our system are that it gives immediate feedback, and that it is scalable because feedback is calculated automatically. 

%There are at least two ways to help students overcome their difficulties with inductive proofs. One way is to adapt the teaching. Several authors have  provided suggestions how to improve education ~\cite{Avital1978, palla, Harel01, ernest}. We will not discuss these improvements in this paper, but concentrate on a second way, namely the development of an e-learning tool which gives the student the opportunity to practice and get feedback, also in cases that no human teacher is available. \Johan{Het een sluit het ander niet uit, natuurlijk. Ik zou dit iets anders `framen'.}

This paper is organized as follows. After introducing our terminology in Section~\ref{terminology}, we discuss related work in Section~\ref{related}. Our research is partly motivated by students' problems with induction, discussed in Section~\ref{problems}.
Section~\ref{interface} describes the interface and functionality of LogInd, and in Section~\ref{hintgeneration} we show how this functionality is realized. A pilot experiment is discussed in Section~\ref{evaluation}, and we conclude in Section~\ref{conclusion} with conclusions and ideas for future work.

\section {Terminology} 
\label{terminology}
We use an example exercise to illustrate the terminology concerning inductive proofs that we use in this paper. The text of the exercise is: 

``The propositional language \ensuremath{\Conid{L}} has atoms \ensuremath{\Varid{p},\Varid{q},\Varid{r}}, ... and connectives \ensuremath{\neg}, \ensuremath{\mathrel{\wedge}} and \ensuremath{\to }. We define two functions on this language: a function \ensuremath{\Varid{prop}} counting all occurrences of propositional letters, and a function \ensuremath{\Varid{bin}} counting the number of binary connectives. These functions are inductively defined by:
\begin{hscode}\SaveRestoreHook
\column{B}{@{}>{\hspre}l<{\hspost}@{}}%
\column{3}{@{}>{\hspre}l<{\hspost}@{}}%
\column{22}{@{}>{\hspre}c<{\hspost}@{}}%
\column{22E}{@{}l@{}}%
\column{25}{@{}>{\hspre}l<{\hspost}@{}}%
\column{E}{@{}>{\hspre}l<{\hspost}@{}}%
\>[3]{}\Varid{prop}\;(\Varid{p}){}\<[22]%
\>[22]{}\mathrel{=}{}\<[22E]%
\>[25]{}\mathrm{1}{}\<[E]%
\\
\>[3]{}\Varid{prop}\;(\neg\phi){}\<[22]%
\>[22]{}\mathrel{=}{}\<[22E]%
\>[25]{}\Varid{prop}\;(\phi){}\<[E]%
\\
\>[3]{}\Varid{prop}\;(\phi\;\Box\;\psi){}\<[22]%
\>[22]{}\mathrel{=}{}\<[22E]%
\>[25]{}\Varid{prop}\;(\phi)\mathbin{+}\Varid{prop}\;(\psi){}\<[E]%
\\[\blanklineskip]%
\>[3]{}\Varid{bin}\;(\Varid{p}){}\<[22]%
\>[22]{}\mathrel{=}{}\<[22E]%
\>[25]{}\mathrm{0}{}\<[E]%
\\
\>[3]{}\Varid{bin}\;(\neg\phi){}\<[22]%
\>[22]{}\mathrel{=}{}\<[22E]%
\>[25]{}\Varid{bin}\;(\phi){}\<[E]%
\\
\>[3]{}\Varid{bin}\;(\phi\;\Box\;\psi){}\<[22]%
\>[22]{}\mathrel{=}{}\<[22E]%
\>[25]{}\Varid{bin}\;(\phi)\mathbin{+}\Varid{bin}\;(\psi)\mathbin{+}\mathrm{1}{}\<[E]%
\ColumnHook
\end{hscode}\resethooks
where \ensuremath{\Varid{p}} is an atom, and \ensuremath{\Box} is \ensuremath{\mathrel{\wedge}} or \ensuremath{\to }. Prove with induction that \ensuremath{\Varid{prop}\;(\phi)\mathrel{=}\Varid{bin}\;(\phi)\mathbin{+}\mathrm{1}} for any formula \ensuremath{\phi} in the language \ensuremath{\Conid{L}}.''

The statement \ensuremath{\Varid{prop}\;(\phi)\mathrel{=}\Varid{bin}\;(\phi)\mathbin{+}\mathrm{1}} in the last sentence is the \textit{theorem} or \textit{property} that has to be proven. The structure of an inductive proof for this theorem can be deduced from the inductive definition of the language \ensuremath{\Conid{L}}. 
The \textit{base case} consists of a proof of the theorem for atomic formulae. There is an \textit{inductive case} for each of the connectives in the language. For instance, a proof of the conjunction case is a proof that from the assumption that if the theorem holds for \ensuremath{\phi} and \ensuremath{\psi} (i.e.~\ensuremath{\Varid{prop}\;(\phi)\mathrel{=}\Varid{bin}\;(\phi)\mathbin{+}\mathrm{1}} and \ensuremath{\Varid{prop}\;(\psi)\mathrel{=}\Varid{bin}\;(\psi)\mathbin{+}\mathrm{1}}) it follows that the theorem also holds for \ensuremath{\phi\mathrel{\wedge}\psi} (i.e.~\ensuremath{\Varid{prop}\;(\phi\mathrel{\wedge}\psi)\mathrel{=}\Varid{bin}\;(\phi\mathrel{\wedge}\psi)\mathbin{+}\mathrm{1}}). The assumption that the theorem holds for some arbitrary formulae \ensuremath{\phi} and \ensuremath{\psi} is the \textit{induction hypothesis}. A \textit{subproof} is a part of the complete inductive proof where a base case or inductive case is proven.

\section{Related work}
\label{related}
Tools that assist a user in constructing structural induction proofs may have different functionalities or purposes. In this section we describe four different kinds of tools: automated theorem provers, proof assistants with didactic functionality, e-learning tools for mathematical induction, and e-learning tools for structural induction.

Bundy~\cite{bundy2001} states that G\"odel's incompleteness theorem implies that it is impossible to construct a completely automatic inductive theorem prover, and that Kreisel's result on cut-elimination~\cite{kreisel} implies that inductive proofs in general will need intermediate lemmas. Hence, automated theorem provers for induction are proof assistants that use several heuristics to try to automatically prove theorems as much as possible. The classical literature on automated theorem proving, for example the handbook of Boyer-Moore~\cite{Boyer98}, describes strategies to find structural induction proofs. This includes different ways of using the induction hypothesis (weak and strong fertilization), selection of the induction variable, and the recognition of the need for extra lemmas. More recent research adds the use of rippling as an important technique~\cite{Bundy2005}. To use an automated theorem prover a user should have a thorough understanding of inductive proofs. Because the concept of induction is the learning goal in courses on proof techniques, using such provers in education is in general not very helpful.
%Since our goal is to teach students the basic concepts of an inductive proof, we only use weak and strong fertilization in our tool. 

Proof assistants such as Tutch~\cite{Abel2001} and Minifn~\cite{Osera} have been developed for educational purposes. Tutch is based on a high-level proving language, which allows step sizes resembling the steps in a pen-and-paper proof. Minifn tries to integrate functional programming with mathematical induction in a way that students can quickly learn how to use the proof assistant. Although these proof assistants are much easier to use than regular theorem provers, they remain assistants. They do not offer exercises, nor do they provide feedback or hints. 

A couple of e-learning tools support learning mathematical induction. In EAsy, an e-assessment system, a student practices with different kinds of proof exercises, using rules and proof strategies from a drop-down menu~\cite{Gruttmann2008}. For a problem requiring a proof by induction, the system pre-structures the proof in a base case and an inductive case, and it provides the induction hypothesis. The system performs a selected rule, which ensures that a student cannot make a mistake in applying a rule. Since there are quite a lot of rules and the exercises are non-trivial, a student easily can get lost: in an evaluation 41\% of the students mentions having problems in selecting the right rule from the extensive ruleset. Completed proofs are graded automatically, but incomplete proofs have to be graded by hand. EAsy shows a student that she completed a proof or subproof successfully, but does not provide any other feedback. 

In the Intelligent Book, a student practices exercises on mathematical induction~\cite{Billingsley2007}. This e-learning system uses MathTiles, predefined templates, which have to be completed by the student. An automated theorem prover, Isabelle, is used in the backend to check correctness and provide simplifications. As in EAsy, the system automatically generates the goals for the different cases. The authors state that this is necessary since Isabelle only accepts these goals when they are exactly equal to the goals in the prover. Simplification can be performed automatically, but when, for example, an application of the induction hypothesis is needed, simplification is not accepted. The tool provides hints in some cases, based on teacher scripts. 

ComIn-M contains electronic exercise sheets~\cite{ComInM}. The system is developed as part of the SAiL-M project, and extends an earlier e-assessment tool~\cite{Mueller2010}. These sheets also contain pre-structured mathematical induction exercises, and combine multiple choice exercises about stating the induction hypothesis with open exercises to prove a base case and an inductive case. Students have to state the inductive case before proving it. A nice feature of ComIn-M is the possibility to work in two directions. Feedback and hints are provided automatically.

The  website\footnote{\url{http://higheredbcs.wiley.com/legacy/college/ensley/0471476021/anim_flash/index.html}} accompanying the textbook Discrete Mathematics: Mathematical Reasoning and Proof with Puzzles, Patterns and Games by D. Ensley and W. Crawley~\cite{deensley} contains a set of interactive exercises in which given a property \ensuremath{\Conid{P}}, a student first has to deduce \ensuremath{\Conid{P}\;(\Varid{n}\mathbin{+}\mathrm{1})} from \ensuremath{\Conid{P}\;(\Varid{n})} for concrete values, and then for an arbitrary value \ensuremath{\Varid{n}}. In this way, a student gets some intuition behind induction before she proves the inductive case.

We found only two systems addressing structural induction. The most extensive system is Polycarpou's e-book, which is a complete educational environment for learning structural induction. She emphasizes foundational concepts, such as structures, sets and closed sets. A separate chapter introduces inductive definitions. Animations show how these definitions generate inductive sets. Interactivity is restricted to multiple choice, drag and drop, and fill in the blank exercises, which implies that a student does not independently complete an inductive proof. 

Stanford's online Introduction to Logic course\footnote{\url{http://intrologic.stanford.edu/public/index.php}} does offer the possibility to construct complete structural induction proofs in its open online logic course. However, in the course material induction is combined with natural deduction, causing long and abstract proofs. It is possible to ask for a complete solution to an exercise, but the system does not give feedback on mistakes, nor hints on how to proceed. The site also offers a proof assistant that combines structural induction with Hilbert-style proofs. 

\section{Students' problems with structural induction}
\label{problems}
Students have problems with constructing inductive proofs. Several studies identify misconceptions students have with mathematical induction, and analyze the underlying reasons for these misconceptions~\cite{Avital1978,  Dubinsky, ernest, Harel01, palla, Pavlekovic}. In her thesis, Polycarpou studies possible causes of problems students have with structural induction~\cite{Polycarpou}. Her hypothesis is that a lack of understanding of set theoretical concepts is one of the main causes of these problems. To investigate this hypothesis, she performed an experiment in which students have to answer six questions about a fancy inductively defined language IPO (a fictive programming
language). The base elements of this language are lower case characters. There are two inductive rules to construct new words: by concatenating two IPO words by an underscore, and by putting quotes around an IPO word. The first four questions test whether a student understands this definition: the student has to indicate which words are IPO words in a given set of words, give the minimal length of an IPO word, determine whether or not IPO words have a maximum length, and construct an IPO word of length greater than 6. The fifth question prepares for the last question by asking if it is possible to construct a word of length 8. Finally, in the last question students have to tell whether an IPO word can have even length, and in case the answer is no, prove that all IPO words have odd length. Students participating in this experiment are enrolled in a course `Logic for Computer Science' and they have practiced with mathematical induction in an earlier course on discrete mathematics. As a pre test, students have to answer these exercises two months before the lessons on induction. After these (traditional classroom) lessons a comparable set of exercises is given as a post test.

In the pre test, students particularly experience problems with the exercise on identifying IPO words, the problem about a word with length 8, and the inductive proof. The results on these questions are much better in the post test, but still the inductive proof is too hard for 44\% of the students. The author claims that there is a correlation between  performance on the first five questions and performance on the inductive proof, but this correlation is not statistically motivated.  Instead, she calculates the ratio between students who perform well on both the inductive set exercises and the inductive proof exercise (71\%), and the ratio between students who do not perform well on both (83\%). She notes that some students find an (incorrect) IPO word with length 8, but manage to prove that all IPO words have odd lengths. Another observation is that some students seem to copy an example treated in the course without the necessary adaptations. These students define inductive cases for negation and conjunction, instead of for the IPO constructors. The conclusion of Polycarpou is that lack of understanding of an inductive definition is indeed a main cause of problems with induction, and that the traditional way of teaching results in procedural knowledge without conceptual understanding for some students.

To investigate whether students entering the master Computer Science at the Open University of the Netherlands, a distance learning university, experience similar problems, we analyze the solutions to a homework assignment of 20 of these students. Before being admitted to the Computer Science master program, students have to take a course in logic. Structural induction is one of the topics in this course. Unlike the students participating in Polycarpou's experiment, almost none of these students has experience with mathematical induction. Furthermore, their background in mathematics is usually weaker than that of bachelor students at a regular university. The homework assignment is optional, but gives extra credits at the exam. 

In the first part of the homework assignment (Exercise 1), students have to give an inductive definition of a function \ensuremath{\Varid{len}}, which returns the length of a propositional formula. The next question (Exercise 2) asks to give an inductive definition of a postfix function \ensuremath{\mathbin{*}} that rewrites all conjunctive subformulae \ensuremath{\phi\mathrel{\wedge}\psi} of a formula into the equivalent subformula \ensuremath{\neg(\neg\phi\mathrel{\vee}\neg\psi)}. The last question (Exercise 3) asks for an inductive proof of the following property: \ensuremath{\Varid{len}\;(\phi\mathbin{*})\leq \mathrm{3}\;\Varid{len}\;(\phi)\mathbin{-}\mathrm{2}} for all formulae \ensuremath{\phi}. This assignment differs from Polycarpou's: it does not test the understanding of inductively defined sets, but instead 
tests the understanding of inductively defined functions (in exercises 1 and 2). 

As Polycarpou, we expect correlations between performance on the first two exercises and the last exercise. Since the number of participating students is too low for a statistic test, we perform a similar calculation as Polycarpou. Students who do not receive full points for the first two exercises (11 students), are almost all (10) unable to complete the third exercise. However, from the students who receive full points (9) only 2 manage to complete the inductive proof. We conclude that for these students Exercise~3 is too difficult to complete without help, but most students have fewer problems with the inductive function definitions. Table~\ref{results1_2} shows the results on the first two exercises. Most students provide a correct definition, but some students add an induction hypothesis to this definition. Since students can make several mistakes, for example, `no correct cases' and `an incorrectly added induction  hypothesis', the sum of the percentages in Table~\ref{results1_2} (and also Table~\ref{results3}) is more than 100\%. 

We looked further into mistakes students made in the proof exercise. As shown in Table~\ref{results3}, the inductive part of the proof most often goes wrong, and the most common error is assigning a fixed length to a compound formula (for example, \ensuremath{\Varid{len}\;(\phi\mathrel{\wedge}\psi)\mathrel{=}\mathrm{3}}). The inductive definition does not seem the bottleneck, since 70\% of the students (see the columns `correct' + `correct use of IH, but incomplete' + `correct ind def, no use of IH' in Table~\ref{results3}) apply these definitions in their proof. The assumption \ensuremath{\phi\mathbin{*}\mathrel{=}\phi} made by some students could be a symptom of conceptual misunderstanding of an inductive definition, but could also be used because a student does not know how to apply the induction hypothesis. We conclude that although problems with inductive definitions may certainly play a role in the results in the inductive proofs, the understanding of the role of the induction hypothesis and the way this hypothesis can be used is perhaps a more important cause of problems. We think that the remedy of Polycarpou (an intelligent tutoring system that mainly focuses on theoretical foundations)  probably will not be the best solution for our students, since the concept of inductive sets does not seem to be their main problem, and her approach might be too theoretical. Instead we concentrate on an e-learning tool that guides students interactively through the construction of an inductive proof.

\begin{table}[t]
    \caption{Results for exercises 1 and 2 of the homework assignment}
    \label{results1_2}
    \begin{subtable}[t]{.5\linewidth}
      \centering
        \caption{Exercise 1}
        \begin{tabular}{lr}
            solution & (N = 20) \\
            \hline
            correct base and ind. cases & 80\% \\
            only correct base case & 5\% \\
            one missing case & 5\% \\
            no correct cases & 10\% \\
            incorrectly added IH & 15\% \\

        \end{tabular}
    \end{subtable}%
    \begin{subtable}[t]{.5\linewidth}
      \centering
        \caption{Exercise 2}
        \begin{tabular}{lr}
           solution & (N = 20) \\
            \hline
            correct & 60\% \\
            no base case & 30\% \\
            incorrect & 10\% \\
        \end{tabular}
    \end{subtable} 
\end{table}

\begin{table}[t]
\caption{Results and mistakes for Exercise 3 of the homework assignment (N = 20)}
\centering
\label{results3}
\begin{tabular}{lrrr}
solution&base&IH&induction\\
\hline
correct & 75\% & 50\% & 15\% \\
IH only for one formula \ensuremath{\phi}& -- & 20\% & \hspace*{1cm}-- \\
assumption \ensuremath{\phi} and \ensuremath{\psi} atomic & \hspace*{1cm}-- & \hspace*{1cm}-- & 35\% \\
assumption \ensuremath{\phi\mathbin{*}\mathrel{=}\phi} & -- & -- & 15\% \\
correct use of IH, but incomplete & -- & -- & 10\% \\
correct ind def, no use of IH & -- & -- & 45\% \\
IH as goal & -- & -- & 10\% \\
verbal intuitive argument & -- & -- & 5\% \\

\end{tabular}
\end{table}

\section{LogInd, a tool for teaching structural induction}
\label{interface}
This section describes LogInd, a tool that supports students with constructing inductive proofs. Experience with other intelligent tutoring systems for logic (LogEx for rewriting propositional formulae~\cite{jcal}, and LogAx for Hilbert-style axiomatic proofs~\cite{logax}) shows that students benefit from a system where they can enter solutions stepwise, get feedback after each step, and can ask for a hint or next step at any moment, or receive a worked-out solution. The possibility to add proof steps both backwards and forwards in these systems resembles the way an exercise is solved with pen and paper. We use the same approach in LogInd. Since some students have hardly any idea how to start an inductive proof, LogInd offers guidance in structuring the proof. We also want students to be aware of the kind of steps they perform in the proof, for example, applying the induction hypothesis or an inductive definition of a given function. Hence, LogInd asks for a justification at each step. Students do not have to justify simple calculation steps, such as distributing a multiplication over an addition, and we allow calculation steps at different levels of granularity. Therefore, LogInd checks calculations by normalizing the submitted expression.

LogInd guides a student through a proof by structuring the proof in three parts: a proof of the base case, stating the induction hypotheses, and a proof of the inductive cases. After presenting the exercise, LogInd asks the student first to state what is to be proven in the base case, see Figure~\ref{start}. 
If this is correct the student is asked to complete the proof of the base case, and to continue with stating the induction hypotheses, see Figure~\ref{base}. For the inductive cases, LogInd again first asks what the different cases are, and what has to be proven in these cases. A complete proof is shown in Figure~\ref{complete}. 
\begin{figure*}[t]
\center\includegraphics[width=11.8 cm]{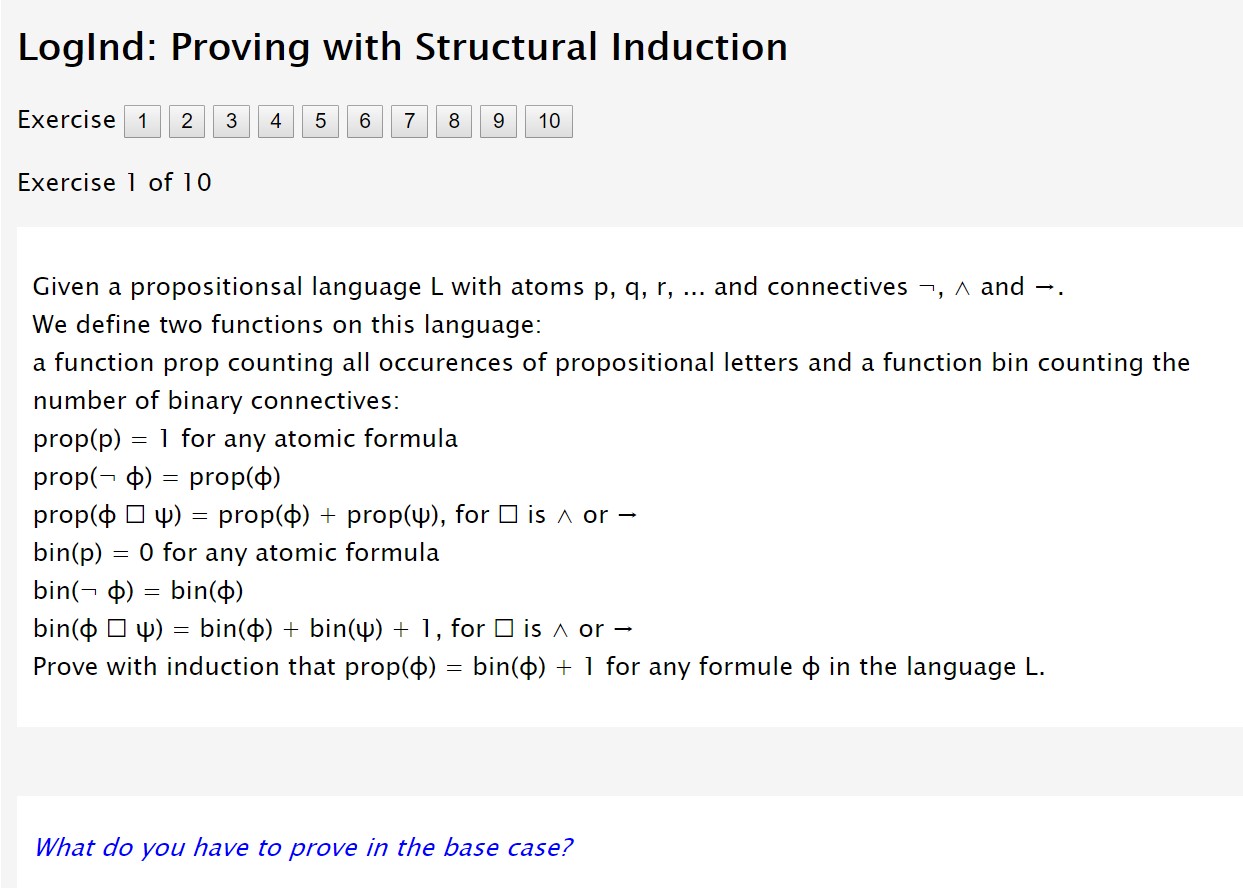}
\caption{Starting the first exercise in LogInd} 
\label{start} 
\end{figure*}

\begin{figure*}[t]
\center\includegraphics[width=6cm]{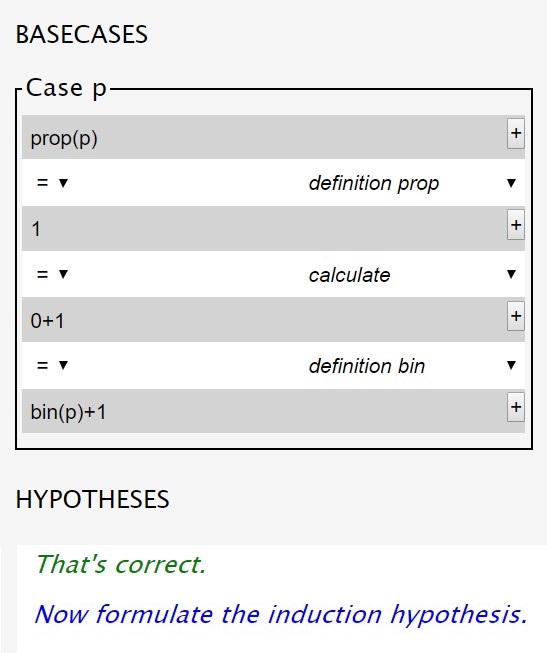}

\caption{Guidance after the base case is finished} 
\label{base} 
\end{figure*}
\begin{figure*}[t]
\center\includegraphics[width=14.15cm]{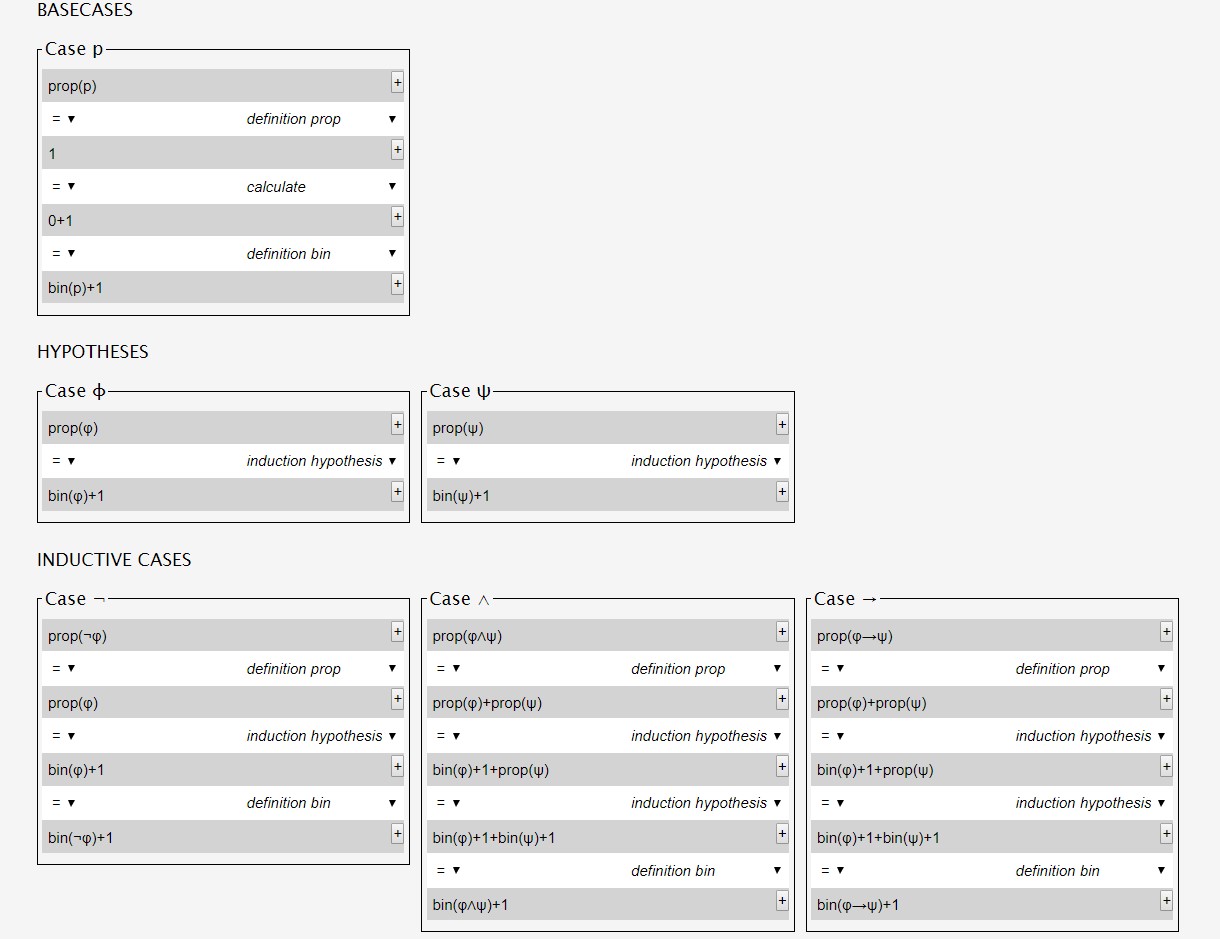}

\caption{The solution of the exercise of Figure~\ref{start} in LogInd} 
\label{complete} 
\end{figure*}

LogInd uses a domain reasoner to provide hints, next steps, feedback, and complete solutions. A domain reasoner is an expert module that performs all reasoning about the domain~\cite{Goguadze}. Thus far, LogInd only offers exercises about properties of a propositional language. We introduce the term `counting function' to describe the set of exercises that can be used in LogInd.

A counting function \ensuremath{\Varid{count}} is an inductively defined function such that
\begin{hscode}\SaveRestoreHook
\column{B}{@{}>{\hspre}l<{\hspost}@{}}%
\column{3}{@{}>{\hspre}l<{\hspost}@{}}%
\column{23}{@{}>{\hspre}l<{\hspost}@{}}%
\column{82}{@{}>{\hspre}l<{\hspost}@{}}%
\column{E}{@{}>{\hspre}l<{\hspost}@{}}%
\>[3]{}\Varid{count}\;(p_i){}\<[23]%
\>[23]{}\mathrel{=}c_i,{}\<[82]%
\>[82]{}c_i\;\in\;\mathbb N ,{}\<[E]%
\\
\>[3]{}\Varid{count}\;(\neg\phi){}\<[23]%
\>[23]{}\mathrel{=}\Varid{a}\mathbin{+}\Varid{b}\cdot\Varid{count}\;(\phi),{}\<[82]%
\>[82]{}\Varid{a},\Varid{b},\in\;\mathbb N {}\<[E]%
\\
\>[3]{}\Varid{count}\;(\phi\;\Box\;\psi){}\<[23]%
\>[23]{}\mathrel{=}a_{\Box}\mathbin{+}b_{\Box}\cdot\Varid{count}\;(\phi)\mathbin{+}c_{\Box}\cdot\Varid{count}\;(\psi),\quad{}\<[82]%
\>[82]{}a_{\Box},b_{\Box},c_{\Box}\;\in\;\mathbb N {}\<[E]%
\ColumnHook
\end{hscode}\resethooks
where \ensuremath{p_i} is a propositional letter, and \ensuremath{\Box} a binary connective. 

The properties that have to be proven take the following form: \ensuremath{P_1\;(\phi)} \textit{comp} \ensuremath{P_2\;(\phi)}, where \textit{comp} is a comparator (\ensuremath{\mathrel{=},\mathbin{<},\leq ,\mathbin{>},\geq }) and \ensuremath{P_i\;(\phi),\Varid{i}\mathrel{=}\mathrm{1},\mathrm{2}} is either a truth value, number or formula:
\begin{itize}
\item if \ensuremath{P_i\;(\phi)} is a truth value, it is an expression \ensuremath{\Conid{V}\;(\Varid{g}\;(\phi))} or  a constant where \ensuremath{\Conid{V}} is a valuation and \ensuremath{\Varid{g}} an inductive  function from the language \ensuremath{\Conid{L}} to \ensuremath{\Conid{L}};
\item if \ensuremath{P_i\;(\phi)} is a natural number, the right-hand side is a linear combination of expressions \ensuremath{\Varid{f}\;(\Varid{g}\;(\phi))} where \ensuremath{\Varid{f}} is a counting function and \ensuremath{\Varid{g}} an inductively defined function from \ensuremath{\Conid{L}} to \ensuremath{\Conid{L}}, the left-hand side is a single expression \ensuremath{\Varid{f}\;(\Varid{g}\;(\phi))};
\item if \ensuremath{P_i\;(\phi)} is a a formula, it is equal to an expression \ensuremath{\Varid{g}\;(\phi)} where \ensuremath{\Varid{g}} is an inductively defined function from \ensuremath{\Conid{L}} to \ensuremath{\Conid{L}};
\item valuations \ensuremath{\Conid{V}} may have predefined properties such as \ensuremath{\Conid{V}\;(\Varid{p})\mathrel{=}\mathrm{1}}, but if the comparator in the theorem is an equality, these properties may also only make use of equalities (since in our proof system it is not possible to prove an equality from inequalities). 
\end{itize}

\noindent
The restriction in the second option that the left-hand side is a single term \ensuremath{\Varid{f}\;(\Varid{g}\;(\phi))} while the right-hand side may contain a linear combination of terms is not a real restriction, since a statement where both the left- and right-hand side contain linear expressions can always be rewritten into this form. The restriction will enable us to treat the induction hypothesis as a rewrite rule. Examples of exercises in this format are:
\begin{itize}
\item Let \ensuremath{\Conid{L}} be a propositional language with connectives \ensuremath{\mathrel{\wedge}} and \ensuremath{\mathrel{\vee}}. Let \ensuremath{\Conid{ValA}} and \ensuremath{\Conid{ValB}} be two valuations such that \ensuremath{\Conid{ValA}\;(\Varid{p})\leq \Conid{ValB}\;(\Varid{p})} for any atomic formula \ensuremath{\Varid{p}}. Then \ensuremath{\Conid{ValA}\;(\phi)\leq \Conid{ValB}\;(\phi)} for any formula \ensuremath{\phi} in the language \ensuremath{\Conid{L}}.
\item Let \ensuremath{\Conid{L}} be a propositional language with connectives \ensuremath{\mathrel{\wedge}} and \ensuremath{\mathrel{\vee}}, and \ensuremath{\Conid{L'}} the extension of \ensuremath{\Conid{L}} with negation \ensuremath{\neg}. The function \ensuremath{\Varid{star}} from \ensuremath{\Conid{L}} to \ensuremath{\Conid{L'}} replaces every atom by its negation, conjunctions by disjunctions and disjunctions by conjunctions. The function \ensuremath{\Varid{length}} returns the length of a formula. The following holds: \ensuremath{\Varid{length}\;(\Varid{star}\;(\phi))\leq \mathrm{2}\cdot\Varid{length}\;(\phi)}
\item Let \ensuremath{\Conid{L}} be a propositional language with connectives \ensuremath{\neg}, \ensuremath{\mathrel{\wedge}}, \ensuremath{\mathrel{\vee}} and \ensuremath{\to }.  Function \ensuremath{\Varid{f}} replaces every conjunctive subformula \ensuremath{\phi\mathrel{\wedge}\psi} by \ensuremath{\neg(\neg\phi\mathrel{\vee}\neg\psi)}, and function \ensuremath{\Varid{g}} replaces every implicative subformula \ensuremath{\phi\to \psi} by \ensuremath{\neg\phi\mathrel{\vee}\psi}. Then \ensuremath{\Varid{f}\;(\Varid{g}\;(\phi))\mathrel{=}\Varid{g}\;(\Varid{f}\;(\phi))} (where `='  means syntactically equal) for any formula \ensuremath{\phi}.
\end{itize}
This class of problems offers sufficient possibilities for relatively simple exercises, where students can get acquainted with inductive proofs. In the next section we show that for this class we can generate solutions, hints and next steps, without the need for advanced techniques as used by automatic theorem provers.

\section{Generation of solutions, hints and next steps}
\label{hintgeneration}
In general, automatic proof generation for induction problems is undecidable~\cite{Aubin, bundy2001}. Problems that might seem easy, such as the proof for associativity of list concatenation for a single list 
(\mbox{\ensuremath{(\Varid{l}\mathbin{::}\Varid{l})\mathbin{::}\Varid{l}}} = \mbox{\ensuremath{(\Varid{l}\mathbin{::}\Varid{l})\mathbin{::}\Varid{l}}}), already need advanced methods to be proven automatically (in this case generalization of the first occurrence of \ensuremath{\Varid{l}})~\cite{bundy2001}. Automatic theorem provers use sophisticated techniques such as rippling and lemma generation to solve inductive problems~\cite{Bundy2005, bundy2001}. It is not our goal to teach students these techniques, and by restricting ourselves to the class of problems described in the previous section, we only need a straightforward strategy to solve such exercises.  

The problem-solving strategy we use is part of our domain reasoner. The strategy first uses the definition of the language to decide what has to be proven in the base cases, and which inductive cases have to be treated. We only allow a single formula variable in a property, so we do not have to ask which variable will be used for induction. Inductive functions are represented as rewrite rules, just as the induction hypothesis. Apart from some technical details, the strategy first applies the inductive definitions of the functions occurring in the statement. The strategy rewrites both the right-hand and left-hand side of the statement. 
Hence, the strategy supports the possibility to complete a subproof by working in two directions. In case the inductive function has natural numbers as the codomain, the next step (if necessary) is a distribution of the multiplication such that the left-hand side and right-hand side become linear combinations of terms that occur in the induction hypothesis. Now we can apply the induction hypothesis to occurrences of the left-hand side of this hypothesis in the left-hand side formula in the proof. After this application only some normalizing elementary calculations might be needed to complete the proof. For our restricted class of problems this strategy always finds a solution. We provide a sketch of a proof of this statement in Appendix~\ref{appendix}. Students will not always follow this strategy. For example, they might apply the induction hypothesis before rewriting the right-hand side of a statement applying the inductive definitions, or vary in the calculations. Also in these cases, LogInd continues with applying the strategy. Hence, a next step can be provided as long as the strategy's rules are applicable. After application of these rules, normalization is sufficient to complete the proof. We expect that almost all student steps will be such that LogInd can indeed  provide a hint or next step based on the student solution. The evaluation described in Section~\ref{evaluation} gives evidence for this claim.

Our strategy differs from the method advocated by Bundy~\cite{bundy2001}. The difference is in the way we use the induction hypothesis. Bundy recommends strong fertilization: the isolation of the induction hypothesis in an inductive case and replacement of this hypothesis by True. Our strategy uses weak fertilization: substituting the left-hand side of the induction hypothesis by the right-hand side. Bundy advises the use of strong fertilization since the use of weak fertilization generally results in longer and more complicated proofs. For our restricted class of problems, this is not the case. Since most pen-and-paper proofs use weak fertilization, we also use weak fertilization in our strategy.

\section{Constraints and feedback}
From the analysis of the homework assignment, we expect our students to make various kinds of mistakes while practicing with LogInd. Examples of potential mistakes are:
\begin{itize}
\item treating metavariables \ensuremath{\phi} and \ensuremath{\psi} as atoms, for example, resulting in the rewriting of \ensuremath{\Varid{length}\;(\phi\mathrel{\wedge}\psi)} into 3 in the proof of an inductive case;
\item omission of a case, for example negation;
\item use of only \ensuremath{\mathrel{=}} and \ensuremath{\leq } when a statement \ensuremath{P_1\;(\phi)\mathbin{<}P_2\;(\phi)} has to be proven;
\item forgetting to state the induction hypothesis before using it.
\end{itize}

In our other tutoring systems for logic we use buggy rules to generate feedback in case a student makes a mistake. Buggy rules typically relate to mistakes on the level of single steps. An example of such a buggy rule is forgetting to change a disjunction in a conjunction in an application of DeMorgan while rewriting a formula into normal form. Mistakes in an inductive proof can be on the level of a step, for instance rewriting \ensuremath{\Varid{length}\;(\phi\mathrel{\wedge}\psi)} to 3, but also on the level of a subproof. An example of an error in a subproof is using \ensuremath{\leq } between each of the steps when the goal is to prove an equality. In this case, each of the steps is correct, but the overall relation between the first and last line of the proof is \ensuremath{\leq } instead of \ensuremath{\mathrel{=}}. An example of an error on the level of the whole inductive proof is the omission of a case. These mistakes are easily formulated as constraint violations. For example, the composition of the relations between the lines in a proof should imply the relation between left-hand side and right-hand side in the theorem that is proven. We think that constraints can be put to good use for this domain.

Ohlsson~\cite{Ohlsson} first described the role of constraints in learning. Mitrovic used the concept in the development of an SQL tutor~\cite{Mitrovic2012} and many other tutors. Constraints characterize correct solutions by providing a relevance condition and a satisfaction condition: if the relevance condition holds, the solution should satisfy the satisfaction condition. Some important reasons for using constraints in the development of tutoring systems are that constraints partially play the role of buggy rules, the construction of which is very time consuming, and that constraints can also be used to give feedback if student solutions diverge from model solutions or are partial~\cite{Mitrovic2012, Mitrovic2007}.

LogInd uses constraints to provide feedback and guidance. Heeren and Jeuring~\cite{feedbackservicesjournal} describe the diagnose service used by a domain reasoner to provide feedback. Figure~\ref{fig:diagnose} is based on this description, and shows how we incorporate constraints in the diagnosis. Our diagnose service receives a (partial) student solution and checks whether or not this submission violates a set of constraints. We divide the constraints in constraints on the level of steps, on the level of subproofs, and on the level of proofs. 

\begin{itize}
\item Constraints on the level of steps check:
\begin{itize}
\item[--] if the rewriting of a line (for example the application of an inductive definition) is correct;
\item[--] if the induction hypothesis is applied correctly;
\item[--] if comparators (=, \ensuremath{\mathbin{<}} ..) are used correctly in a single step;
\item[--] if the justification is correct.
\end{itize}

\item Constraints on the level of subproofs check:
\begin{itize}
\item[--] if the first and last line of each subproof are instances of the left-hand side respectively right-hand side of the theorem;
\item[--] if these instantiations are valid cases (atomic base cases, inductive cases only for connectives in the language);
\item[--] if the induction hypothesis (when present) is correctly formulated;
\item[--] if comparators (=, \ensuremath{\mathbin{<}}, ..) are used correctly at the level of subproofs (i.e.~if the composition of the comparators in the subproof equals the comparator in the theorem).
\end{itize}

\item Constraints on the level of complete proofs check:
\begin{itize}
\item[--] if all inductive cases correspond to connectives in the language;
\item[--] if the induction hypothesis is stated before use, with the same metavariables;
\item[--] if an inductive case or special base case is missing. 
\end{itize}
\end{itize}

To check a proof at the level of a step, the domain reasoner compares the student submissions with the result of the application of possible rules, and accepts the student submission if it is similar to one of these results. Here, a simple normalizing calculation transforms the submission and the generated result into the same expression. For example, application of the induction hypotheses in the example of Figure~\ref{complete} results in \ensuremath{\Varid{bin}\;(\phi)\mathbin{+}\mathrm{1}\mathbin{+}\Varid{bin}\;(\psi)\mathbin{+}\mathrm{1}}, but a submission \ensuremath{\Varid{bin}\;(\phi)\mathbin{+}\Varid{bin}\;(\psi)\mathbin{+}\mathrm{2}} is also accepted.
 
If no constraint is violated, the domain reasoner compares the new submission with the previous (last correct) submission, and determines if these are similar. If these submissions are similar, the domain reasoner gives feedback about this. A submission that is not similar may follow the implemented strategy, which is diagnosed as `expected'. When the step does not follow the strategy, but is recognized by the domain reasoner, the diagnosis is a `detour'. This happens, for example, when a student starts with completing the inductive case for implication before negation. The interface will tell the student that this step is correct, and the student can continue with the exercise. Since there are no violations, every step in the submission is already recognized, which means that the last option (no rule detected) only happens when the student submission contains more than one new line. Again, the interface will provide a message `this is correct'. 
 
Our experience is that the diagnosis `failure' of a constraint is not enough to provide informative feedback. For example, a constraint on the exercise in Figure~\ref{start} could be that the first line of an inductive case should be an instantiation of the left-hand side of the theorem \ensuremath{\Varid{prop}\;(\phi)\mathrel{=}\Varid{bin}\;(\phi)\mathbin{+}\mathrm{1}}, with \ensuremath{\phi} substituted by \ensuremath{\neg\alpha}, \ensuremath{\alpha\mathrel{\wedge}\beta} or \ensuremath{\alpha\to \beta}, where \ensuremath{\alpha\;\neq\;\beta} and \ensuremath{\alpha,\beta\;\in\;\{\mskip1.5mu \phi,\psi\mskip1.5mu\}}. Now a student can violate this constraint in different ways, for example by
\begin{itize}
\item using a connective that is not in the language;
\item instantiating with an atomic formula;
\item instantiating with a metavariable that is not used in the induction hypothesis;
\item introducing an expression that is not an instantiation at all.
\end{itize}  
Each of these violations stem from a different misconception, and we want to give different feedback messages in each case. We solve this by specifying failure messages for different constraints,
and call this use of constraints `buggy constraints' as proposed by Kodaganallur et al.~\cite{Kodaganallur}. 
One of the problems of the use of constraints as mentioned by Kodagallur et al.~\cite{Kodaganallur} and Mitrovic \cite{Mitrovic2012} is the violation of two or more constraints at the same time. We solve this by ordering the constraints: for example, constraints about instantiations get a higher priority than constraints about the application of a rule.

Apart from `strong' constraints that may not be violated, we also use soft constraints to guide a student through a proof. These constraints check for each of the subproofs if they are introduced and if they are finished. After a diagnosis, the user interface can call the feedback service `constraints', which reports the status of each of the subproofs. The result can be used to provide a message such as: `the base case is finished, continue with the formulation of the induction hypothesis'.

\begin{figure}[t]
\begin{center}
\begin{tikzpicture}
\decision{(0,2.5)}{violation \\ constraint?}
\decision{(3,2.5)}{buggy \\constraint?}
\decision{(0,0)}{similar?}
\decision{(3,0)}{expected by \\ strategy?}
\decision{(6,0)}{discover \\ rule?}
\diagnosis{(6.4,2.5)}{Unknown mistake}
\diagnosis{(5.5,1.4)}{Common mistake}
\diagnosis{(0,-1.7)}{Small rewrite step,\\ not recognized}
\diagnosis{(3,-1.7)}{Rewrite step follows\\ expert strategy }
\diagnosis{(9.2,0)}{Multiple steps }
\diagnosis{(8.5,-1.1)}{Correct step, but\\ detour from strategy }
\draw[->] (1,2.5) -- (2,2.5); \draw (1.5,2.65) node {\smalltext{yes}};
\draw[->] (4,2.5) -- (5,2.5); \draw (4.5,2.65) node {\smalltext{no}};
\draw[->] (1,0) -- (2,0); \draw (1.5,0.15) node {\smalltext{no}};
\draw[->] (4,0) -- (5,0); \draw (4.5,0.15) node {\smalltext{no}};
\draw[->] (7,0) -- (7.8,0); \draw (7.4,0.15) node {\smalltext{no}};
\draw[->] (3,2) -- (3,1.4) -- (4.1,1.4); \draw[right] (3,1.7) node {\smalltext{yes}};
\draw[->] (6,-.5) -- (6,-1.1) -- (7.1,-1.1); \draw[right] (6,-.8) node {\smalltext{yes}};
\draw[->] (0,-.5) -- (0,-1.3); \draw[right] (0,-.9) node {\smalltext{yes}};
\draw[->] (3,-.5) -- (3,-1.3); \draw[right] (3,-.9) node {\smalltext{yes}};
\draw[->] (0,2) -- (0,.5); \draw[right] (0,1.25) node {\smalltext{no}};
\draw[->] (0,3.8) -- (0,3); \draw[right] (0,3.6) node {\smalltext{diagnose\\feedback service}};
\end{tikzpicture}
\end{center}
\caption{Structure of the diagnose feedback service: the incorporation of constraints is new}
\label{fig:diagnose}
\end{figure}
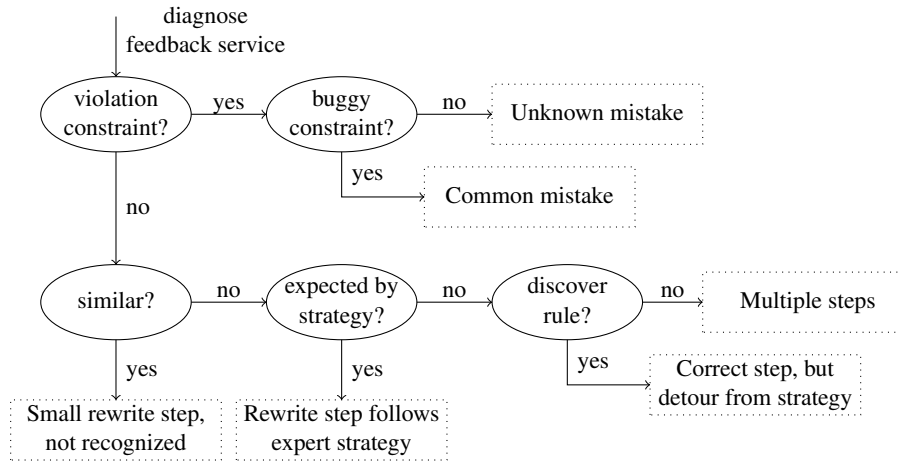

The way we use constraints in LogInd differs in some aspects from the original use. One difference is the fact that LogInd has a strategy which produces solutions, and while checking the next step of a student, LogInd first tries to recognize this step. LogInd can hence be conceived as a constraint-based solver as described in Kodaganallur et al.~\cite{Kodaganallur}. Moreover, we use constraints not only to provide feedback on errors, but also to guide a student.  

\section{Evaluation}
\label{evaluation}
In April 2019 we performed a small pilot experiment with a group of 15 students taking an online logic course in preparation of admission to the master in Computer Science at the Open University of the Netherlands. Before the experiment, these students handed in the homework assignment described in Section~\ref{problems}. The experiment consisted of an online instruction about the use of LogInd, followed by the possibility to practice. During the experiment students could ask questions by using the chat functionality of the learning environment.
The questions that we would like to answer with this experiment are:
\begin{enumerate}
\item does LogInd behave as expected, i.e., provide hints and next steps, and give a correct diagnosis?
\item what kind of problems do students have while working with LogInd?
\item how do students use LogInd?
\item can we see effects of the use of LogInd in the way students perform pen-and-paper exercises?
\end{enumerate}  

In Section~\ref{hintgeneration} we claim that we can provide a hint or next step at any moment, also if a student does not follow the preferred strategy by LogInd. We analyze the loggings to check this claim. From the 1612 calls obtained from student interactions (a diagnosis, a hint, a next step, or a full solution), 
398 calls ask for a hint or a next step. In 25 cases (6\%) LogInd cannot provide such a step. Further analysis shows that since students repeat their call several times, this only happens in five different cases (1\%). In all of these cases LogInd could not provide a hint or next step because it diagnoses the exercise as `ready'. This was a bug, resulting from a use of LogInd that we had not anticipated: some students did not start a base case with the statement that they have to prove, but they submit the first two lines of a proof of the base case, for example the subproof:
\begin{hscode}\SaveRestoreHook
\column{B}{@{}>{\hspre}l<{\hspost}@{}}%
\column{3}{@{}>{\hspre}l<{\hspost}@{}}%
\column{E}{@{}>{\hspre}l<{\hspost}@{}}%
\>[3]{}\Varid{prop}\;(\Varid{p}){}\<[E]%
\\
\>[3]{}\mathrel{=}(\Varid{definition}\;\Varid{prop}){}\<[E]%
\\
\>[B]{}\mathrm{1}{}\<[E]%
\ColumnHook
\end{hscode}\resethooks
In such a case LogInd decided that this step was correct and that this subproof was finished since all steps are motivated, without checking whether indeed the base case has been proven. When a student asked for a hint after the other subproofs were (correctly) finished, LogInd could not provide such a hint. We repaired this bug, and since this was the only reason that no hint or next step was available, we expect that LogInd now indeed provides this kind of feed forward (i.e. hints and next steps) in all circumstances. 

We use the remarks made by students in the chat during the experiment and the loggings to answer the second question. From these remarks and the loggings we learned that students had quite a lot of problems with the interface. Students should start an exercise with a response to the question `what do you have to prove in the base case?'. They should fill in their answer in a template as shown in Figure~\ref{startbase}. For the first exercise, the exercise of the example in Section~\ref{terminology}, this means that on the first line, a student should enter \ensuremath{\Varid{prop}\;(\Varid{p})}, then choose the equality sign (=) from the drop-down list and enter the right-hand side \ensuremath{\Varid{bin}\;(\Varid{p})\mathbin{+}\mathrm{1}} in the bottom line. In their first attempt, none of the participating students entered this first step correctly. Ten students did not realize that they first had to answer this question before completing the proof or did not know what to prove. They provided answers such as for example \ensuremath{\Varid{prop}\;(\Varid{p})\mathrel{=}\mathrm{1}} or \ensuremath{\Varid{prop}\;(\phi)\mathrel{=}\Varid{bin}\;(\phi)\mathbin{+}\mathrm{1}}. Three students entered the whole statement \ensuremath{\Varid{prop}\;(\Varid{p})\mathrel{=}\Varid{bin}\;(\Varid{p})\mathbin{+}\mathrm{1}} on the first line, one student added a wrong justification (definition of the inductively defined function \ensuremath{\Varid{prop}}), and one student replaced the ? by an empty string, which at that time was not accepted by LogInd. 

A second source of problems was the use of the send button. A student only gets feedback after clicking this button. If a student enters several lines before clicking this button, it was hard to find the place where the feedback referred to, especially since at the time of the experiment we were still developing the constraints. So a student might receive the message `this step is not correct' without a clue which step should be repaired. 
Also, when a student asked for a hint after receiving an error message, this hint was based on the latest correct submission. Hence, this hint might relate to the application of a rule in (for example) a base case, while the student was working on an inductive case.

\begin{figure*}[t]
\center\includegraphics[width=9 cm]{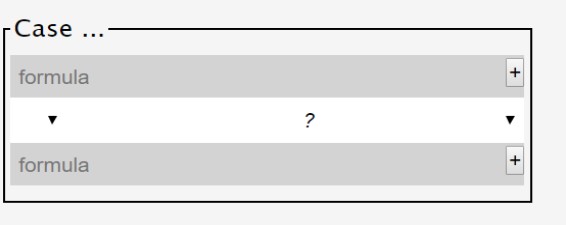}

\caption{Template for starting a base case} 
\label{startbase} 
\end{figure*}

To answer the third question, how do students use LogInd, we also looked at the loggings. The problems described in the previous section combined with the conceptual and technical problems students have with induction caused a high hint and next step use (25\% of the student interactions in the loggings). However, students did enter steps themselves and asked for a diagnosis (1078 requests, 67\%), where 619 steps were diagnosed as correct. Half of the participating students were able to construct (parts of) a subproof, some of them using hints or next steps in between, but the other students had too many problems with the interface. Students who solved the homework assignment before the experiment, could use LogInd without too many problems. 

To answer the last question, we analyzed resubmissions of students' homework. Students whose homework assignment was not correct had to submit the assignment again. We hoped to use these improved assignments as a kind of post test. However, three students submitted their corrections before the experiment, and one student did not submit a new version. The results of the remaining seven students can be found in Table~\ref{results3post}. We use the same characterizations of solutions as in Table~\ref{results3}, except for the solution  label `incorrect use of IH',
which did not occur in the first submissions. The first column shows the results on the first submissions of the homework assignment by the group of students who practiced with LogInd and submitted a new version after the experiment. The results of this second submission is shown in the second column. The last two columns show the same data for the group of students who did not practice with LogInd. The second submission is in both groups better than the first attempt. The number of students is too low to conclude if practicing with LogInd has more effect than the comments by the teacher on the first attempt. 

\begin{table}[t]
\begin{center}
\begin{tabular}{lrrrr}
Solution          & \multicolumn{2} {c} {LogInd (N = 7)} & \multicolumn{2}{c} {No LogInd (N = 4)}  \\ 
                 & 1st  & 2nd                            & 1st   & 2nd                      \\ \hline
Correct       & \hspace*{1cm}-- & 2 & \hspace*{1cm}-- & \hspace*{1cm}-- \\
Assumption \ensuremath{\phi} and \ensuremath{\psi} atomic & 3 & 1 & 4 & 1 \\
Assumption \ensuremath{\phi\mathbin{*}\mathrel{=}\phi} & 1 & 3 & 0 & 1 \\
Correct use of IH but incomplete & 1 & \hspace*{1cm}-- & -- & -- \\
Incorrect use of IH & -- & 1 & -- & -- \\
Correct ind def, no use of IH & 3 & 5 &  4& 2 \\

\end{tabular}
\end{center}
\caption{Results and mistakes in the first submission of homework assignment 3 and in the improved second submission (number of students)}
\label{results3post}
\end{table}

\section{Conclusion and future work}
\label{conclusion}
We have discussed the design of a tutoring system for learning how to prove statements about inductively defined discrete structures. As far as we know, LogInd is the first such tutoring system. A student constructs her proof stepwise, and LogInd provides help (hints, next steps, and elaborate feedback on errors) at each step. A pilot evaluation showed that LogInd indeed can provide help in almost all situations. Half of the students in the experiment could construct (parts of) a proof by themselves, but the other half had too many problems with the interface and the exercises. The number of students who participated in the experiment was too low to decide whether students indeed learn by using LogInd. We noticed that homework submissions by students who had practiced with LogInd were more clearly structured, and contained more justifications of different steps.

In a next experiment we will evaluate the constraints: is the information in a feedback message correct, and do these messages help students to correct their submission? We will test an alternative interface, where students can enter their steps in a text field, and we will also compare a guided version with a non-guided version. In the future we might also incorporate exercises about induction in other domains, such as for example lists or functional programs.

\section* {Acknowledgments} 
We thank Aad van Lieburg and Anja Paalvast
for their work on the student 
interface, and our students  for 
their permission to use their solutions in our research, and their participation in the experiment.

\begin{small}
\bibliographystyle{eptcs}
\bibliography{strategies,FeedbackServices,induction}
\end{small}

\newpage
\appendix
 \addcontentsline{toc}{section}{Appendices}
 
\section{Sketch of a completeness proof for the strategy used by LogInd}
\label{appendix}
We give a sketch of the completeness proof for the strategy used by LogInd. 
First we specify the class of functions that are used in in our exercises for transforming a formula in another formula. We call these inductively defined functions acceptable.
%an be proved, we will sketch a proof for the class of theorems comparing counting functions of formulae that are possibly transformed by an inductively defined function. Before we can sketch this part of the completeness proof, we need to specify . 
 
\newtheorem{mydef}{Definition}

\begin{mydef}
An inductively defined function \ensuremath{\Varid{g}} from the propositional language \ensuremath{\Conid{L}} to \ensuremath{\Conid{L}} is acceptable if the inductive definition can be written in the following form:
\begin{itize}
\item \makebox[1.8cm][l]{\ensuremath{\Varid{g}\;(p_i)}}           \ensuremath{\mathrel{=}\phi_i} \quad\quad for atomic formula \ensuremath{p_i}
\item \makebox[1.8cm][l]{\ensuremath{\Varid{g}\;(\neg\phi)}}       \ensuremath{\mathrel{=}[\mskip1.5mu \Varid{g}\;(\phi)\mathbin{/}\Varid{s}\mskip1.5mu]\;\psi}
\item \makebox[1.8cm][l]{\ensuremath{\Varid{g}\;(\phi_1\;\Box\;\phi_2)}} \ensuremath{\mathrel{=}[\mskip1.5mu \Varid{g}\;(\phi_1)\mathbin{/}s_1,\Varid{g}\;(\phi_2)\mathbin{/}s_2\mskip1.5mu]\;\psi_{\Box}}
\end{itize}
% > g(p_i)            =  phii
% > g(not phi)        =  [g(phi)/s]psi
% > g(phi1 box phi2)  =  [g(phi1)/s_1, g(phi2)/s_2]psi_box
% for atomic formula |p_i|.
\end{mydef}

\vspace*{2mm}
The definition states that the inductive cases are obtained by substituting a variable \ensuremath{\Varid{s}} by \ensuremath{\Varid{g}\;(\phi)}  in a formula \ensuremath{\psi}, or the variables \ensuremath{s_1} and \ensuremath{s_2} by  \ensuremath{\Varid{g}\;(\phi_1)} and \ensuremath{\Varid{g}\;(\phi_2)} in a formula \ensuremath{\psi_{\Box}}.
All inductive definitions of functions from \ensuremath{\Conid{L}} to \ensuremath{\Conid{L}} can be written this way. For example, the function \ensuremath{\Varid{star}} in the second example in Section~\ref{interface} can be defined by:
\begin{hscode}\SaveRestoreHook
\column{B}{@{}>{\hspre}l<{\hspost}@{}}%
\column{3}{@{}>{\hspre}l<{\hspost}@{}}%
\column{20}{@{}>{\hspre}c<{\hspost}@{}}%
\column{20E}{@{}l@{}}%
\column{23}{@{}>{\hspre}l<{\hspost}@{}}%
\column{E}{@{}>{\hspre}l<{\hspost}@{}}%
\>[3]{}\Varid{g}\;(p_i){}\<[20]%
\>[20]{}\mathrel{=}{}\<[20E]%
\>[23]{}\neg\;p_i{}\<[E]%
\\
\>[3]{}\Varid{g}\;(\neg\phi){}\<[20]%
\>[20]{}\mathrel{=}{}\<[20E]%
\>[23]{}[\mskip1.5mu \Varid{g}\;(\phi)\mathbin{/}\Varid{s}\mskip1.5mu]\;(\neg\Varid{s}){}\<[E]%
\\
\>[3]{}\Varid{g}\;(\phi_1\mathrel{\wedge}\phi_2){}\<[20]%
\>[20]{}\mathrel{=}{}\<[20E]%
\>[23]{}[\mskip1.5mu \Varid{g}\;(\phi_1)\mathbin{/}s_1,\Varid{g}\;(\phi_2)\mathbin{/}s_2\mskip1.5mu]\;(s_1\mathrel{\vee}s_2){}\<[E]%
\\
\>[3]{}\Varid{g}\;(\phi_1\mathrel{\vee}\phi_2){}\<[20]%
\>[20]{}\mathrel{=}{}\<[20E]%
\>[23]{}[\mskip1.5mu \Varid{g}\;(\phi_1)\mathbin{/}s_1,\Varid{g}\;(\phi_2)\mathbin{/}s_2\mskip1.5mu]\;(s_1\mathrel{\wedge}s_2){}\<[E]%
\ColumnHook
\end{hscode}\resethooks

In our proof we need the following lemma:
\newtheorem{lemma}{Lemma}
\begin{lemma}
\label{substitution}
For any counting function \ensuremath{\Varid{f}} and formulae \ensuremath{\phi} and \ensuremath{\psi},  
\ensuremath{\Varid{f}\;([\mskip1.5mu \phi\mathbin{/}\Varid{p}\mskip1.5mu]\;\psi)} is a linear expression in \ensuremath{\Varid{f}\;(\phi)}.
\end{lemma}
\begin{proof}
Proof with induction on \ensuremath{\psi}:\\
Since \ensuremath{\Varid{f}} is a counting function, there exists constants \ensuremath{c_i}, \ensuremath{\Varid{a}}, \ensuremath{\Varid{b}}, \ensuremath{a_{\Box}}, \ensuremath{b_{1\Box}}, and \ensuremath{b_{2\Box}} such that
\begin{hscode}\SaveRestoreHook
\column{B}{@{}>{\hspre}l<{\hspost}@{}}%
\column{3}{@{}>{\hspre}l<{\hspost}@{}}%
\column{21}{@{}>{\hspre}c<{\hspost}@{}}%
\column{21E}{@{}l@{}}%
\column{24}{@{}>{\hspre}l<{\hspost}@{}}%
\column{E}{@{}>{\hspre}l<{\hspost}@{}}%
\>[3]{}\Varid{f}\;(p_i){}\<[21]%
\>[21]{}\mathrel{=}{}\<[21E]%
\>[24]{}c_i,{}\<[E]%
\\
\>[3]{}\Varid{f}\;(\neg\phi){}\<[21]%
\>[21]{}\mathrel{=}{}\<[21E]%
\>[24]{}\Varid{a}\mathbin{+}\Varid{b}\cdot\Varid{f}\;(\phi),{}\<[E]%
\\
\>[3]{}\Varid{f}\;(\phi_1\;\Box\;\phi_2){}\<[21]%
\>[21]{}\mathrel{=}{}\<[21E]%
\>[24]{}a_{\Box}\mathbin{+}b_{1\Box}\cdot\Varid{f}\;(\phi_1)\mathbin{+}b_{2\Box}\cdot\Varid{f}\;(\phi_2){}\<[E]%
\ColumnHook
\end{hscode}\resethooks
For atomic formulae \ensuremath{\psi}, \ensuremath{\Varid{f}\;([\mskip1.5mu \phi\mathbin{/}\Varid{p}\mskip1.5mu]\;\psi)} is either \ensuremath{\Varid{f}\;(p_i)} (a constant) or \ensuremath{\Varid{f}\;(\phi)} and hence linear in \ensuremath{\Varid{f}\;(\phi)}.

\par\noindent
For the inductive cases we assume that \ensuremath{\Varid{f}\;([\mskip1.5mu \phi\mathbin{/}\Varid{p}\mskip1.5mu]\;\psi_1)\mathrel{=}\mbox{$c_1$}\mathbin{+}\mbox{$d_1$}\cdot\Varid{f}\;(\phi)} and \ensuremath{\Varid{f}\;([\mskip1.5mu \phi\mathbin{/}\Varid{p}\mskip1.5mu]\;\psi_2)\mathrel{=}c_2\mathbin{+}d_2\cdot\Varid{f}\;(\phi)}.

\par\noindent
Then:
\begin{hscode}\SaveRestoreHook
\column{B}{@{}>{\hspre}l<{\hspost}@{}}%
\column{3}{@{}>{\hspre}c<{\hspost}@{}}%
\column{3E}{@{}l@{}}%
\column{6}{@{}>{\hspre}l<{\hspost}@{}}%
\column{E}{@{}>{\hspre}l<{\hspost}@{}}%
\>[6]{}\Varid{f}\;([\mskip1.5mu \phi\mathbin{/}\Varid{p}\mskip1.5mu]\;(\neg\psi_1)){}\<[E]%
\\
\>[3]{}\mathrel{=}{}\<[3E]%
\>[6]{}\Varid{f}\;(\neg[\mskip1.5mu \phi\mathbin{/}\Varid{p}\mskip1.5mu]\;\psi_1){}\<[E]%
\\
\>[3]{}\mathrel{=}{}\<[3E]%
\>[6]{}\Varid{a}\mathbin{+}\Varid{b}\cdot\Varid{f}\;([\mskip1.5mu \phi\mathbin{/}\Varid{p}\mskip1.5mu]\;\psi_1){}\<[E]%
\\
\>[3]{}\mathrel{=}{}\<[3E]%
\>[6]{}\Varid{a}\mathbin{+}\Varid{b}\cdot(\mbox{$c_1$}\mathbin{+}\mbox{$d_1$}\cdot\Varid{f}\;(\phi)){}\<[E]%
\\
\>[3]{}\mathrel{=}{}\<[3E]%
\>[6]{}\Varid{a}\mathbin{+}\Varid{b}\cdot\mbox{$c_1$}\mathbin{+}\Varid{b}\cdot\mbox{$d_1$}\cdot\Varid{f}\;(\phi){}\<[E]%
\ColumnHook
\end{hscode}\resethooks
And:\begin{hscode}\SaveRestoreHook
\column{B}{@{}>{\hspre}l<{\hspost}@{}}%
\column{3}{@{}>{\hspre}c<{\hspost}@{}}%
\column{3E}{@{}l@{}}%
\column{6}{@{}>{\hspre}l<{\hspost}@{}}%
\column{22}{@{}>{\hspre}l<{\hspost}@{}}%
\column{27}{@{}>{\hspre}l<{\hspost}@{}}%
\column{45}{@{}>{\hspre}l<{\hspost}@{}}%
\column{63}{@{}>{\hspre}l<{\hspost}@{}}%
\column{81}{@{}>{\hspre}l<{\hspost}@{}}%
\column{E}{@{}>{\hspre}l<{\hspost}@{}}%
\>[6]{}\Varid{f}\;([\mskip1.5mu \phi\mathbin{/}\Varid{p}\mskip1.5mu]\;(\psi_1\;\Box\;\psi_2)){}\<[E]%
\\
\>[3]{}\mathrel{=}{}\<[3E]%
\>[6]{}\Varid{f}\;([\mskip1.5mu \phi\mathbin{/}\Varid{p}\mskip1.5mu]\;\psi_1\;{}\<[22]%
\>[22]{}\Box\;[\mskip1.5mu \phi\mathbin{/}\Varid{p}\mskip1.5mu]\;\psi_2){}\<[E]%
\\
\>[3]{}\mathrel{=}{}\<[3E]%
\>[6]{}a_{\Box}\mathbin{+}b_{1\Box}\cdot\Varid{f}\;([\mskip1.5mu \phi\mathbin{/}\Varid{p}\mskip1.5mu]\;\psi_1)\mathbin{+}b_{2\Box}\cdot\Varid{f}\;([\mskip1.5mu \phi\mathbin{/}\Varid{p}\mskip1.5mu]\;\psi_2){}\<[E]%
\\
\>[3]{}\mathrel{=}{}\<[3E]%
\>[6]{}a_{\Box}\mathbin{+}b_{1\Box}\cdot(\mbox{$c_1$}\mathbin{+}\mbox{$d_1$}\cdot\Varid{f}\;(\phi))\mathbin{+}b_{2\Box}\cdot(c_2\mathbin{+}d_2\cdot\Varid{f}\;(\phi)){}\<[E]%
\\
\>[3]{}\mathrel{=}{}\<[3E]%
\>[6]{}a_{\Box}\mathbin{+}b_{1\Box}\cdot{}\<[27]%
\>[27]{}\mbox{$c_1$}\mathbin{+}b_{2\Box}\cdot{}\<[45]%
\>[45]{}c_2\mathbin{+}(b_{1\Box}\cdot{}\<[63]%
\>[63]{}\mbox{$d_1$}\mathbin{+}b_{2\Box}\cdot{}\<[81]%
\>[81]{}d_2)\cdot\Varid{f}\;(\phi){}\<[E]%
\ColumnHook
\end{hscode}\resethooks
\end{proof}

In the same way we can prove that for any counting function \ensuremath{\Varid{f}} and formulae \ensuremath{\phi_1}, \ensuremath{\phi_2} and \ensuremath{\psi},  
\ensuremath{\Varid{f}\;([\mskip1.5mu \phi_1\mathbin{/}s_1,\phi_2\mathbin{/}s_2\mskip1.5mu]\;\psi)} is a linear expression in \ensuremath{\Varid{f}\;(\phi_1)} and \ensuremath{\Varid{f}\;(\phi_2)}. We omit the proof. Using Lemma~\ref{substitution} we can prove the next lemma:

\begin{lemma}
\label{linear}
For any counting function \ensuremath{\Varid{f}} and acceptable inductively defined function \ensuremath{\Varid{g}}, after applying \ensuremath{\Varid{g}}
\begin{itize}
\item \ensuremath{\Varid{f}\;(\Varid{g}\;(p_i))} is a constant for atomic formulae \ensuremath{p_i}
\item \ensuremath{\Varid{f}\;(\Varid{g}\;(\neg\phi))} is a linear expression in \ensuremath{\Varid{f}\;(\Varid{g}\;(\phi))}
\item \ensuremath{\Varid{f}\;(\Varid{g}\;(\phi_1\;\Box\;\phi_2))} is a linear expression in \ensuremath{\Varid{f}\;(\Varid{g}\;(\phi_1))} and \ensuremath{\Varid{f}\;(\Varid{g}\;(\phi_2))}
\end{itize}
\end{lemma}
\begin{proof}

For atomic formulae,  \ensuremath{\Varid{g}\;(p_i)} is a propositional formula, and hence, \ensuremath{\Varid{f}\;(\Varid{g}\;(p_i))} is a constant. \\
Since \ensuremath{\Varid{g}} is acceptable, there exists a formula \ensuremath{\psi} and variable \ensuremath{\Varid{s}} such that \ensuremath{\Varid{f}\;(\Varid{g}\;(\neg\phi))\mathrel{=}\Varid{f}\;([\mskip1.5mu \Varid{g}\;(\phi)\mathbin{/}\Varid{s}\mskip1.5mu]\;\psi}), which is linear in \ensuremath{\Varid{f}\;(\Varid{g}\;(\phi))} according to Lemma~\ref{substitution}. \\
In the same way: there exists a formula \ensuremath{\psi_{\Box}} and variables \ensuremath{s_1} and \ensuremath{s_2} such that \ensuremath{\Varid{f}\;(\Varid{g}\;(\phi_1\;\Box\;\phi_2))\mathrel{=}\Varid{f}\;([\mskip1.5mu \Varid{g}\;(\phi_1)\mathbin{/}s_1,\Varid{g}\;(\phi_2)\mathbin{/}s_2\mskip1.5mu]\;\psi_{\Box})},  which is linear in \ensuremath{\Varid{f}\;(\Varid{g}\;(\phi_1))} and \ensuremath{\Varid{f}\;(\Varid{g}\;(\phi_2))} by the generalization of Lemma~\ref{substitution}.
\end{proof}

\newtheorem{thm}{Theorem}
\begin{thm}
The strategy used by LogInd can generate an inductive proof for any correct statement of the form \ensuremath{P_1\;(\phi)}   \textit{comp} \ensuremath{P_2\;(\phi)},  where \textit{comp} is a comparator (\ensuremath{\mathrel{=},\mathbin{<},\leq ,\mathbin{>},\geq }), \ensuremath{P_2\;(\phi)} is a linear combination of expressions \ensuremath{f_i\;(g_i\;(\phi))}, and \ensuremath{P_1\;(\phi)} is a single expression \ensuremath{\Varid{f}\;(\Varid{g}\;(\phi))}, for which \ensuremath{\Varid{f}} and \ensuremath{f_i}  are  counting functions, and \ensuremath{\Varid{g}} and \ensuremath{g_i} are inductively defined functions from \ensuremath{\Conid{L}} to \ensuremath{\Conid{L}}.
\end{thm}
\begin{proof}
The strategy starts with the application of the inductively defined functions. We first consider the base case. After the application of inductively defined functions, the left-hand side of the equation is a constant and the right-hand side a linear combination of constants, which is rewritten into a single constant using arithmetic, by Lemma~\ref{linear}. A number comparison suffices to check if the base case holds.
In the inductive cases, the application of the inductively defined function in the left-hand side results in a linear  expression in \ensuremath{\Varid{f}\;(\Varid{g}\;(\phi))} (case negation) or \ensuremath{\Varid{f}\;(\Varid{g}\;(\phi_1))} and \ensuremath{\Varid{f}\;(\Varid{g}\;(\phi_2))} (case binary connective). The right-hand side is a linear combination of linear expressions in \ensuremath{f_i\;(g_i\;(\phi))} or in \ensuremath{f_i\;(g_i\;(\phi_1))} and \ensuremath{f_i\;(g_i\;(\phi_2))}. The next step in the algorithm is the application of the induction hypothesis. Replacing occurrences of \ensuremath{\Varid{f}\;(\Varid{g}\;(\phi))} or \ensuremath{\Varid{f}\;(\Varid{g}\;(\phi_1))} and \ensuremath{\Varid{f}\;(\Varid{g}\;(\phi_2))} by the right-hand side of the induction hypothesis results in another linear combination of \ensuremath{f_i\;(g_i\;(\phi))} or \ensuremath{f_i\;(g_i\;(\phi_1))} and \ensuremath{f_i\;(g_i\;(\phi_2))}. Normalizing both left-hand side and right-hand side now suffices to prove the inductive cases.
\end{proof}

\end{document}